\documentclass[10pt]{article}
\usepackage{latexsym, amscd, amsfonts, eucal, mathrsfs, amsmath, amssymb, amsthm, xr,  makeidx, stmaryrd}
\usepackage[letterpaper, portrait, margin=1in]{geometry}

\usepackage[usenames]{color}
\usepackage{verbatim}
\usepackage{tensor}
\usepackage{enumitem}
\usepackage[colorlinks=true,allcolors=blue]{hyperref}
\usepackage{epigraph}
\usepackage{graphicx}
\usepackage{tikz}

\usepackage{verbatim}
\usepackage{tensor}
\usepackage{enumitem}
\usepackage{subcaption}
\usetikzlibrary{backgrounds}
\usepackage{float}
\newtheorem{prop}{Proposition}
\newtheorem*{claim}{Claim}

\newtheorem{theorem}{Theorem}

\usepackage[english]{babel}
\usepackage[english=british]{csquotes}
\usepackage[comma,round,authoryear,longnamesfirst]{natbib}

\usepackage{turnstile}
\usepackage{cancel}
\usepackage{amssymb}

\usepackage{mathtools}
\usepackage{stackengine}
\setstackEOL{\\}
\usepackage{tikz}
\usepackage{tikz-3dplot}
\usepackage{tikz-cd}
\usetikzlibrary{arrows,decorations.pathmorphing,backgrounds,positioning,fit,petri,calc,shapes.misc,decorations.markings,intersections}

\usepackage{amsmath,amsfonts,amsthm} 
\usepackage{stmaryrd}
\usepackage{enumitem}
\usepackage{caption}
\usepackage{array}

\begin{document}

\title{Probabilism for Stochastic Theories}
\author{Jeremy Steeger\\
  \texttt{jsteeger@nd.edu}}
\date{\today}
\maketitle

\begin{abstract}
I defend an analog of probabilism that characterizes rationally coherent estimates for chances. Specifically, I demonstrate the following accuracy-dominance result for stochastic theories in the $C^*$-algebraic framework: supposing an assignment of chance values is possible if and only if it is given by a pure state on a given algebra, your estimates for chances avoid accuracy-dominance if and only if they are given by a state on that algebra. When your estimates avoid accuracy-dominance (roughly: when you cannot guarantee that other estimates would be more accurate), I say that they are sufficiently coherent. In formal epistemology and quantum foundations, the notion of rational coherence that gets more attention requires that you never allow for a sure loss (or `Dutch book') in a given sort of betting game; I call this notion full coherence. I characterize when these two notions of rational coherence align, and I show that there is a quantum state giving estimates that are sufficiently coherent, but not fully coherent.
\end{abstract}

\tableofcontents


\section{Introduction}\label{sec:intro}

It is quite well understood how to justify the rules of probability as rational constraints on uncertainty---at least when `uncertainty' is understood as a given agent's degrees of belief in true-or-false propositions. Such a justification is provided by any one of the numerous arguments for (classical) probabilism, the thesis that degrees of belief (i.e. credences) are rationally coherent if and only if they follow the rules of classical probability theory.
\begin{quote}
    \textbf{Classical probabilism.} Credences are rationally coherent if and only if they are given by a classical probability function.
\end{quote}
Arguments for classical probabilism---including the popular `Dutch book' argument---generally proceed by giving some precise characterization of what it means for degrees of belief to rationally cohere, and then showing that all and only classical probability functions meet the formally-stated desiderata.

Things are less clear-cut when talking about other sorts of uncertainty. In the sciences, agents often make uncertain claims that cannot be straightforwardly reduced to credence-talk. Some of these claims are \emph{estimates}, i.e. approximate calculations of the values of variables. Estimates may be calculated in a variety of ways, and they need not match a possible assignment of values.

For example, imagine we have two identical-looking coins, coin $X$ and coin $Y$. Suppose that, after many tosses, we agree that the outcomes for flips of these coins have the following chances: a three-fourths chance of heads and a one-fourth chance of tails for flips of coin $X$, and the converse chances of heads and tails for flips of coin $Y$. Now suppose I shuffle the coins behind my back and give you one of them. There are two \emph{possible assignments of chance values} for your coin---the assignment for coin $X$ and that for coin $Y$---but since you do not know which coin you got, you do not know the actual chances. Now suppose I ask you to estimate these chances. One intuitive way to do so uses credences: take the weighted average of the possible assignments of chance values, where each weight is determined by your credence in the claim that you received the appropriate coin. Another intuitive method of estimation uses actual frequencies: observe the fraction of heads in a large number of tosses of your coin, and pick the probability function with the chance of heads that is closest to this value without going over three-fourths or under one-fourth. Each of these methods can yield estimates of one-half for the chances of heads and tails. So intuitively, such estimates should rationally cohere. But this intuition raises a question: if estimates for chances need not match possible assignments of chance values, then when, precisely, are they rationally coherent?


I argue that, in some cases, the rules of probability characterize when estimates for chances rationally cohere. In particular, I argue that, in the right setting, there is an analog of classical probabilism that applies to stochastic theories. The `right setting' is the $C^*$-algebraic framework for such theories. $C^*$-algebras are used in quantum mechanics, quantum statistical mechanics, and some approaches to quantum field theory; classical probability spaces are recovered as special cases of such algebras. In this framework, probability functions correspond to algebraic states, i.e. states on a $C^*$-algebra. Some of these states, e.g. pure states, have certain desirable properties. These and other sorts of states will be rigorously defined in Section \ref{sec:coherence}; for finite-dimensional Hilbert spaces, density matrices correspond to algebraic states and wavefunctions correspond to pure states. The thesis of probabilism for stochastic theories asserts that, when possible assignments of chance values are given by pure states, estimates for chances should match some algebraic state:
\begin{quote}
    \textbf{Probabilism for stochastic theories.} Suppose an assignment of chance values is possible if and only if it is given by a pure state. Then estimates for chances are rationally coherent if and only if they are given by an algebraic state.\footnote{This thesis requires the stipulation of values of chances, and so it is not consistent with interpretations of stochastic theories that deny the existence of any sort of objective probabilities---e.g. Bayesian personalists and QBists need not apply.}
\end{quote}
The thesis also holds when `pure state' is replaced with `vector state', `normal state', or `algebraic state' (c.f. Section \ref{sec:coherence} for definitions). To defend this thesis, I borrow a technical notion of the rational coherence of estimates from \cite{deFinetti-ToP:1974}: \emph{sufficient coherence}, or the avoidance of accuracy-dominance.\footnote{My rational coherence arguments, like those of de Finetti, are \emph{synchronic}---i.e. they invoke facts of the matter at some given time. I do not consider updating rules (Bayesian or otherwise). Synchronic arguments that differ from mine are offered by \cite{Fuchs2013a} and \cite{Pitowsky2003}; unlike my arguments, these arguments assume a strongly subjective approach to probability. Diachronic arguments are given by \cite{Wallace2012a} and \cite{Greaves2010a}; unlike my arguments, these arguments assume a many-worlds approach to quantum theory. I leave a more detailed treatment of the relationship between these arguments and my own for future work.}

Sufficient coherence is slightly weaker than the familiar notion of `Dutch book' coherence. `Dutch book' coherence requires (roughly) that you avoid a sure loss in a given sort of betting game; henceforth, I will refer to this notion as \emph{full coherence}. Sufficient coherence requires (roughly) that you cannot get closer to the possible assignments of values by changing your estimates. In other words, your estimates are sufficiently coherent just in case they \emph{avoid accuracy-dominance}---roughly, just in case you cannot guarantee that other estimates would be more accurate (c.f. \citealt{Joyce1998,Williams2012,Williams2012a,Pettigrew2016}). By contrast with full coherence, sufficient coherence only requires (roughly) that you could always get \emph{arbitrarily close} to breaking even in the aforementioned betting game (c.f. \citealt{deFinetti-ToP:1974}). Full coherence works well enough for defenses of classical probabilism. But it turns out to be too strong for at least one version of probabilism for stochastic theories---i.e. an algebraic state may yield estimates for chances that are sufficiently coherent, but not fully coherent.\footnote{Specifically, this may occur when an assignment of chance values is possible if and only if it is given by a vector state; c.f. Proposition \ref{prop:not_protected}.}

Section \ref{sec:prelim} develops a few necessary philosophical preliminaries---namely, the meaning of the terms \emph{chances}, \emph{estimates}, and \emph{rational coherence} in the above statement of probabilism for stochastic theories. Regarding chance: I adopt a very thin notion of objective probability on which a chance value is, at minimum, our best estimate for a relative frequency of an outcome in repeated observations of a given sort; I call this approach \emph{functionalism about chance}. Functionalism about chance does not require a commitment to Lewis's (\citeyear{Lewis1980,Lewis1994}) principal principle, although it naturally accommodates one version of this principle. Regarding estimates: I suppose any method of approximate calculation may yield an estimate, so long as the result is a real number. In particular: estimates may be calculated using credences (i.e. they may be credence-weighted averages), but they need not be. Regarding rational coherence: I argue that this notion is adequately captured by the avoidance of accuracy-dominance---i.e. by sufficient coherence.

Section \ref{sec:coherence} presents four main technical results. First: I give a topological characterization of sufficient coherence in infinite dimensions (c.f. Proposition \ref{prop:suff_res}). Second: supposing that pure, vector, normal, or all states on a given $C^*$-algebra yield the possible assignments of chance values, I show that estimates for chances are sufficiently coherent (i.e. they avoid accuracy-dominance) if and only if they are given by a state on that algebra (c.f. Theorem \ref{thm:gen}).\footnote{My definition of a normal state, which is somewhat non-standard, requires just that a state is a countable convex combination of vector states. This aligns with the usual definition in the context of concrete von Neumann algebras. For more details, c.f. Section \ref{sec:coherence_quantum}.} Third: I note a condition that characterizes when full coherence and sufficient coherence align (extending the work of \citealt{Kuhr2007, Williams2012,Williams2012a}; c.f. Propositions \ref{prop:full_suff} and \ref{prop:full_nec}). Fourth: I show that sufficient coherence is needed for a given application to quantum states (c.f. Proposition \ref{prop:not_protected}). Section 4 concludes with a more detailed summary of these results.



\section{Chances, estimates for chances, and coherence of estimates}\label{sec:prelim}

My strategy for defending probabilism for stochastic theories is to adapt old arguments---the familiar `Dutch book' and accuracy-dominance defenses of classical probabilism---to serve a new purpose. The old arguments are useful because they contain intuitive and precise definitions of the rational coherence of credences. These turn out to yield, with minimal modification, an intuitive and precise definition of the rational coherence of estimates for chances.

But I have introduced three novel elements to the old arguments: chances, estimates, and a rational coherence condition for estimates. So the reader may naturally worry that these elements pose new philosophical problems. In this section, I address this concern by giving a brief account of the notions of chance (Section \ref{sec:prelim_functionalism}) and of estimates and their rational coherence (Section \ref{sec:prelim_coherence}) that I will use in my defense of probabilism for stochastic theories.

\subsection{Estimates and functionalism about chance}\label{sec:prelim_functionalism}

In this section, I would like to sketch a notion of an objective probability value that is very thin, but that still allows for us to be uncertain about such a value. My approach is loosely inspired by Brown's (\citeyear{Brown2005}) and Knox's (\citeyear{Knox2017}) writings on spacetime. Very roughly: these authors entertain the idea that \emph{all there need be} to the concept of a spacetime is that it plays the functional role of defining inertial structure. In this spirit, I suggest that \emph{all there need be} to the concept of an objective probability value is that it plays the functional role of specifying the best \emph{estimate} for a relative frequency---where an estimate is just some approximate calculation, and one that need not be tied to an agent's credence in a claim that some particular event occurs.

This approach to chance is more directly inspired by Feynman's lucid description (in his \emph{Lectures}) of how probabilities are used and tested in the lab. Feynman defines the probability of an event $A$ as the fraction $N_A/N$, where $N_A$ is `our \emph{best} estimate' for the number of times $A$ will occur in $N$ repeated observations (\citeyear[\S6-1]{Feynman1965}). I take a probability, in this sense, to be our best estimate for some relative frequency that we intend to measure, before we actually go out and measure it. By an estimate, I mean the following:
\begin{quote}
    \textbf{Estimate.} Consider a variable whose possible values lie in some subset of the real numbers. An \emph{estimate} for that variable is a real number that approximates its actual value.
\end{quote}
So an agent may estimate one-half for the fraction of heads in three tosses of a coin, even though one-half is not one of its possible values; the agent merely thinks that one-half is a good approximation of its actual value. Such an estimate could be calculated using credences (e.g. it could be a credence-weighted average), but it need not be. As Feynman notes, sometimes the best we can do is to simply observe the number of heads in a large number of tosses, and then use this value as our estimate for the fraction of heads in our next three tosses.

Feynman gives two basic instructions on how to use frequency estimates, which I will call \emph{the predictive norm} and \emph{the inferential norm} for ease of reference:
\begin{quote}
    \textbf{The predictive norm.} If our estimate for a relative frequency is $x$, then we should expect that a measurement of this frequency will be close to $x$.
    
    \textbf{The inferential norm.} If the measurement of a relative frequency is close to $x$, then we should not change an estimate of $x$ for this frequency.\footnote{These norms are inspired in part by \cite{Papineau1997}, who asserts that `objective probability' is the thing that satisfies `the inferential link'---i.e. we use frequencies to estimate it---and `the decision-theoretic link'---i.e. we base `rational choices' on our knowledge of it. But on Feynman's tack, we need not invoke decision theory to spell out how frequency estimates inform predictions, and estimates can be formed in ways other than observation of frequencies.}
\end{quote}
There's a bit of an art to unpacking the meaning of `close'; I will stick with the standard way (which is also Feynman's way) that takes `close to $x$' to mean `no more than two or three (binomial) standard deviations away from $x$'. These norms just describe one standard use of probabilities: to make some prediction about a relative frequency that is then tested in the lab.

Functionalism about chance simply adds to the above that we may agree to designate some value as our \emph{best} estimate for the relative frequency of an outcome \emph{for a given sort of observation}---and for any given instance of an observation, we can be uncertain about its sort, and so we can be uncertain about chance. 
\begin{quote}
    \textbf{Functionalism about chance.} The chance of an outcome is, at minimum, our best estimate for its relative frequency in repeated observations \emph{of a given sort}. 
\end{quote}
To illustrate an example, consider again the two identical-looking coins, coin $X$ and coin $Y$. After flipping those coins many times, we agreed to assign values for the chances of two outcomes---heads and tails---for two sorts of observations---flips of coin $X$ and flips of coin $Y$. When I shuffle the coins and give you one of them, you are uncertain whether flips of your coin are flips of coin $X$ or flips of coin $Y$, and so you are uncertain about the chances of the outcomes for flips of your coin.

For a less colloquial example, take the Born-rule assignment of probabilities to outcomes for a given observation of a quantum system. The functionalist takes this assignment to stand in for a (defeasible) assertion made by all users of quantum theory---namely, the assertion that those are our best frequency estimates for that sort of observation. And an agent may well be unsure about the sort of a given observation---e.g. they may be unsure of which wavefunction should be used to describe the system being measured, or of which projection should be used to describe the measurement---and so they may be unsure of which Born-rule estimates are best.\footnote{Note, too, that the functionalist's tack allows for chances to be inaccessible \emph{in principle}: we can posit that there are observations whose sort we cannot discern. And so, in particular, the functionalist can \emph{also} make sense of the chances ascribed by any non-deterministic hidden variable model underlying quantum theory.}

If we know what sort of observation we are making, then we know the chances of the outcomes, and we use these chances as our frequency estimates (and then proceed to follow the predictive and inferential norms for these estimates). That is, we follow a \emph{chance-estimate norm}:
\begin{quote}
    \textbf{Chance-estimate norm.} If we know that the chance of an outcome is $x$, then our estimate for its relative frequency should be $x$. 
\end{quote}
The functionalist need not attempt to \emph{justify} the predictive, inferential, and chance-estimate norms. That is: they need not attempt to explain \emph{why} a given frequency estimate is `best'. Such an attempt may enter the territory of the usual philosophical analyses of objective probability (e.g. hypothetical frequentism, propensity views, best-systems analyses, and so on). Functionalism does allow for the possibility that chances correspond to objective properties that are `out there in the world', as some of these analyses suggest (and so it is thinner than the `no-theory theory' of \cite{Sober2010} and the pragmatic approach of \cite{Healey2012}, which both deny this possibility). It only commits to the claim that chances are \emph{at least} intersubjective, in the thin sense that multiple subjects agree on what the best frequency estimates are. Only the strong subjectivist (i.e. the Bayesian personalist) refuses to consider even this thin sort of agreement. Such subjectivists---e.g. \cite{deFinetti-ToP:1974}, or the QBists \cite{Fuchs2013a}---need not apply.\footnote{Functionalism about chance may deviate from Feynman's notion of probability on the matter of `objectivity'. Feynman's view is hard to discern: he denies that there is such a thing as a `true' or `correct' probability, as he thinks that this notion cannot be made `logically consistent'; nonetheless, he maintains that probabilities are only subjective `in a sense' (\citeyear[\S6-3]{Feynman1965}). I think this ambiguity has its root in one way that Feynman seems to define probability, which runs as follows:
\begin{quote}
    \emph{By the `probability' of a particular outcome of an observation we mean our estimate for the most likely fraction of a number of repeated observations that will yield that particular outcome.} (\citealt[\S6-1]{Feynman1965})
\end{quote}
If this is a definition, then the qualifier `most likely' poses a problem. As Brown (\citeyear{Brown2011a}) notes, there is a `whiff of circularity'---for what does `most likely' mean if not `most probable'? On Brown's reading of Feynman, the definition escapes circularity because there is no sense in which Feynman's `probabilities' are `objective' in the sense of being `out there in the world'---although they could still be `intersubjective' in the sense I describe.

Functionalism about chance takes a different route. The functionalist does not qualify the fractions being estimated as `most likely'; the fractions being estimated are just any relative frequencies we intend to measure. `Most likely' refers to the binomial distribution, and the functionalist only uses this distribution to assign a rough range of numerical values to the term `close'. So they need not endorse an explanation of chance in terms of natural properties, but they can consistently admit the possibility of one. Here is a (somewhat sophistical) consistent model: a propensity view that asserts, for every (functionalist) chance, the existence of a dispositional property yielding actual frequencies that are distributed (roughly) binomially about that chance.}

Estimates for frequencies need not directly encode anything about an agent's credence in a claim that a particular event occurs---i.e. the chance-estimate norm need not entail any sort of norm linking chance to credence in such a claim.\footnote{For overview of such norms, see \cite{Pettigrew2012}.} However, there is a very natural chance-credence norm that the functionalist can recover if they are willing to get a bit more specific about the predictive norm. To wit: suppose that the chance of $A$ for a given sort of observation is $x$ and that you aim to perform $N$ of these observations (where it may be that $N=1$). Now suppose that the predictive norm instructs you to use the binomial distribution fixed by $x$ to set your credences in claims predicting certain values for the fraction $N_A/N$. Then, the chance-estimate norm and the predictive norm imply what \cite{Ismael2008} calls an \emph{unconditional} version of Lewis's (\citeyear{Lewis1980,Lewis1994}) principal principle:
\begin{quote}
\textbf{Principal principle (unconditional).} If we know that the chance of an outcome is $x$, then our credence in the claim that it occurs should be $x$.
\end{quote}
As \cite{Pettigrew2012} notes, this version of the principal principle is an \emph{externalist} norm: it references a quantity---chance---to which we are not guaranteed access. This reflects the fact that the chance-estimate norm is externalist in precisely this way.

To briefly review: for a variable whose possible values lie in some subset of the reals, an \emph{estimate} for that variable is a real number that approximates its actual value. Estimates can be calculated in many ways, and they need not invoke credences. Functionalism takes chance to be, at minimum, our agreed-upon \emph{best} frequency estimate \emph{for a given sort of observation}---and since we can be uncertain about the sort of an observation, we can be uncertain about chance. The usual norms for frequency estimates recommend---at the \emph{very} least---that when we know the best estimate, we should use it. And if we take a more specific approach to these norms, we recover an unconditional version of Lewis's principal principle.

Once we have taken this account of chance on board, we allow for the possibility that agents may not have access to chances. Neither the chance-estimate norm nor the unconditional principal principle instruct agents on what to do when they are uncertain about chance. This task is taken up by probabilism for stochastic theories.

\subsection{Estimates for chance and coherence of estimates}\label{sec:prelim_coherence}

In the previous section, I sketched a thin notion of objective probability---functionalism about chance---that makes chances variables with objective (i.e. at least intersubjective), real-number values. On this view, agents know the possible values of a chance, but they may be uncertain about its actual value. Thus, they could use any of the usual tools for expressing this uncertainty. In particular, they could express credences in claims assigning various values to this chance. Or they could calculate an \emph{estimate} for the chance value---literally, an estimate for our best frequency estimate.

Consider, again, our weighted coins $X$ and $Y$. You may be equally confident that I gave you coin $X$ and that I gave you coin $Y$, in which case you may use a credence-weighted average to calculate your estimates for the chances of heads and tails for flips of your coin. But you may reasonably wish to withhold credence. In this case, you may you may prefer the second method of estimation noted in the introduction: toss the coin one hundred times, note the resulting fraction of heads, and use the closest probability function with a chance of heads between three-fourths and one-fourth to fix your estimates for the chances of the outcomes.

So we have at least two reasonable methods for calculating an estimate for chance, and either of these methods may yield estimates that do not match a possible assignment of chance values. This immediately raises the question: what \emph{constrains} reasonable estimates for chances---that is, precisely when are such estimates \emph{rationally coherent}?\footnote{\label{footnote:GPP}If the reader is partial to the unconditional principle principal, they may also ask: what constrains our credence in the face of uncertainty about chance? An answer is provided by Ismael's (\citeyear{Ismael2008}) \emph{generalized principal principle} (GPP), which instructs an agent to match their credences to a credence-weighted average of possible assignments of chance values (specifically, a countable convex combination of such assignments). My imminent arguments for probabilism for stochastic theories may afford a novel defense of (a slightly modified form of) GPP---but I will not develop such a defense here. See also footnote \ref{footnote:GPP_norm}.} I will refer to this as \emph{the functionalist's question}:
\begin{quote}
    \textbf{The functionalist's question.}  When are estimates for chances \emph{rationally coherent}?
\end{quote}
Probabilism for stochastic theories provides one compelling answer to this question: in some situations (e.g. when an assignment of chance values is possible if and only if it is given by a pure state), estimates for chances are rationally coherent if and only if they are given by a state on a $C^*$-algebra.\footnote{Note that this question is \emph{not} answered by de Finetti's representation theorem (\citealt{deFinetti-ToP:1974,Greaves2010a}) or its quantum analog (\citealt{Caves2002}). These representation theorems show, roughly, that an exchangeable credence can be recovered as a probabilistic mixture of objective probability values; they do not show that such mixtures satisfy some rational coherence condition.}

To lightly motivate this answer, note that it recovers textbook orthodoxy concerning one use of density operators in quantum theory---namely, their use in expressing ignorance about the true state of a system. Density operators are often defined as probabilistic mixtures of one-dimensional projections fixed by wavefunctions. When these mixtures are used to express ignorance about which wavefunction describes a system, these operators are referred to as \emph{proper mixtures}, to contrast them with improper mixtures---i.e. density operators that arise from a description of a subsystem of some larger system described by a wavefunction. The orthodoxy concerning proper mixtures dates back to von Neumann, who uses a density operator to describe a system when `we do not know what state is actually present', where the `actually present' state is described by a wavefunction (\citeyear[pp. 295--296]{Neumann1955}). One modern iteration of this idea is a certain perspective on the use of maximum likelihood estimation (MLE) in quantum state tomography, which restricts `estimates for true states' to density operators (even when non-positive operators would reproduce observed frequencies); we may suppose that `true states' are fixed by wavefunctions to recover something like von Neumann's idea.\footnote{See \cite{Blume-Kohout2010} for a critical introduction to this method. I leave a more detailed discussion of the role of rational coherence arguments in state tomography for future work.} As noted, for finite-dimensional quantum systems, density operators correspond to algebraic states on the relevant algebra (i.e. the bounded operators on a finite-dimensional Hilbert space), and wavefunctions correspond to pure states on that algebra. So probabilism for stochastic theories gives an answer to why density operators should be used as `estimates for true states'---estimates for chances rationally cohere just in case they are given by some density operator. Unfortunately, we do not have a \emph{defense} of this answer.

You could reasonably question the need for a defense. You may think that it is self-evident that estimates for chances should always be calculated as credence-weighted averages. In this case, you may wish to \emph{stipulate} that estimates rationally cohere just in case they arise from such an average. Then the aptness of density operators is built into their definition---since, as noted, they are often \emph{defined} as probabilistic mixtures.\footnote{\label{footnote:GPP_norm}This point can be expanded into an argument that credences should match normal states, using GPP (c.f. footnote \ref{footnote:GPP}): suppose that an assignment of chance values is possible if and only if it is given by a vector state on $\mathcal{B(H)}$; note that all countable convex combinations of such assignments extend to normal states, and that all normal states restrict to such combinations (c.f. \citealt[Theorem 7.1.12]{KadisonRingrose:vol2}); suppose credences should be given by some such combination, i.e. suppose GPP; thus, credences should be given by a normal state. For an approach to this claim that uses an internalist version of Lewis's principal principle, see \cite[Section 5.1]{Arageorgis2017}.} But presently, I will not assume that estimates must be calculated in this way.

Moreover, you may take the Kolmogorovian axioms to be obvious principles constraining both rational credences and rational estimates for chances. Then you may not need a defense of either classical probabilism or its stochastic analog. To wit:  Gleason's theorem and its several algebraic variants (\citealt{Gleason1957,Busch2003,Redei2007}) establish that in nearly all cases, probabilities arising from appropriate states on $C^*$-algebras are all and only those functions following appropriate analogs of Kolmogorov's rules. But presently, I will not take Kolmogorov's rules to be self-evident.

My strategy, then, for giving a defense of probabilism for stochastic theories is to first give a \emph{general} answer to the functionalist's question---an answer that does not assume Kolmogorov's rules and that does not require any particular method of calculating estimates. This strategy is adapted from accuracy-dominance arguments for classical probabilism, which arguments similarly do not depend on Kolmogorovian axioms or a given way of forming credences (\citealt{Joyce1998,Williams2012,Williams2012a,Pettigrew2016}). I begin with the intuition that the goodness of your estimates turns on how accurate they are. Then I suggest that your estimates rationally cohere just in case you could not use your knowledge of what values are possible to guarantee that other estimates would be more accurate. This allows for a technical definition of the rational coherence of estimates: sufficient coherence.

Accuracy-dominance arguments have their roots in the version of the `Dutch book' argument that appears in de Finetti's \emph{Theory of Probability} (\citeyear{deFinetti-ToP:1974}), which has two felicitous features that are often overlooked. First: it uses a condition of rational coherence that applies to \emph{estimates} for \emph{any real-valued `random' variables}---where the variables are only `random' in the sense that their actual values are unknown to an agent. De Finetti refers to an estimate as a `prevision' or a `foresight'; only when an estimate approximates the truth value of a true-or-false proposition does this notion fully align with the familiar notion of credence. In the previous section, I briefly noted that de Finetti dismisses any objective notion of probability, even the thin intersubjective notion proposed by the functionalist. He could not make sense of an `estimate for chance'---chances, for him, could never be the right sort of variables. But, if we take functionalist's tack, we merely \emph{assert} that they are the right sort of variables. And so de Finetti's rational coherence condition applies just as well to them.

Second: the argument uses sufficient coherence to define the rational coherence of estimates, and this notion has both conceptual and technical advantages over full coherence. As noted, an agent's estimates are sufficiently coherent when they could not surely be more accurate. Conceptually: this notion affords a definition of rational coherence that reflects the goodness of estimates as estimates, i.e. as approximations of a value, and it avoids the well-known issues with a pragmatic approach to `Dutch books'. Technically: this notion is needed for one version of probabilism for stochastic theories, as it turns out that there is an algebraic state that may be sufficiently but not fully coherent.

To precisely define sufficient coherence, first consider some set of real-valued random variables. Every assignment of values to these variables is a function from the set of variables to the real numbers. Let any such function that represents a possible state of affairs be a \emph{possible point}.
\begin{quote}
    \textbf{Possible points.} Let $\mathcal{Q} \subseteq \mathbb{R}^K$ be the set of \emph{possible points}; an element of $\mathcal{Q}$ is a function $q: K\to \mathbb{R}$, which is a point in the space $\mathbb{R}^K$.
\end{quote}
Clearly, $\mathcal{Q}$ is a subset of $\mathbb{R}^{K}$, the set of all functions from $K$ to $\mathbb{R}$. I will switch back and forth between talking about $\mathbb{R}^{K}$ as a set of functions and as a topological vector space (i.e., the $K$-fold Cartesian product of $\mathbb{R}$ equipped with the usual product topology) for convenience. In order to compare estimates to the actual values of variables, we also express our estimates with functions from the set of variables to the real numbers, i.e. with points in $\mathbb{R}^K$.
\begin{quote}
    \textbf{Estimate function.} An agent's estimates for a set $K$ of real-valued variables are given by an \emph{estimate function} $e \in \mathbb{R}^K$.
\end{quote}
Now we may give a technical definition of the rational coherence of estimates: sufficient coherence. (For ease of reference, I say that estimates are rationally/sufficiently/fully coherent when the estimate function that assigns them is rationally/sufficiently/fully coherent.) De Finetti gives two equivalent definitions of sufficient coherence in terms of betting games---one of which lightly modifies the usual notion of a `Dutch book'. We will return to this approach in Section \ref{sec:coherence_games} to relate sufficient coherence to full coherence. For now, I will take a slightly different tack: I will define sufficient coherence as the avoidance of accuracy-dominance (a definition which is nonetheless mathematically equivalent to de Finetti's).

The definition proceeds in two parts. First, let us agree to measure the relative closeness of two estimate functions to a possible point with Euclidean distance. Certainly, it is a substantive philosophical assumption that this is a good way to measure such relative distances. I will not defend this assumption at length, but I will note that it reflects common practice in both statistics and epistemology. In particular, the Brier score is a strictly increasing function of Euclidean distance, and so one estimate function has a lower Brier score (i.e. is more accurate) than another just in case it is closer in Euclidean distance to the relevant point.\footnote{See \cite[Chapters 3--4]{Pettigrew2016} for further discussion of and motivation for the Brier score as a measure of the inaccuracy of credence.}  There is a small technical hurdle: since we have not placed any restrictions on the number of variables, the set $K$ could be possibly uncountably infinite, and so we cannot, in general, use the Euclidean metric to compare points in $\mathbb{R}^K$. Fortunately, we can still compare the closeness of estimate functions to the values of any arbitrary finite number of those variables, without having to deviate from the familiar Euclidean notion.

Second, say an estimate function is \emph{sufficiently coherent} just in case no other estimate function is closer to all the possible points over some finite set of variables---that is, you cannot guarantee that some other estimate function is more accurate based solely on your knowledge of the possible points.
\begin{quote}
    \textbf{Sufficient coherence.} An estimate function $e$ is \emph{sufficiently coherent} if and only if it \emph{avoids accuracy-dominance}---that is, there is no other $e^*\in\mathbb{R}^K$ such that, for some finite subset $\Gamma$ of $K$,
    \begin{equation*}\label{eq:suff_acc}
        \text{for all } q\in\mathcal{Q}, \quad     d(q|_\Gamma,e|_\Gamma) - d(q|_\Gamma,e^*|_\Gamma)  > 0,
    \end{equation*}
    where $d$ is Euclidean distance.
\end{quote}

Sufficient coherence is a \emph{non-pragmatic} constraint: it simply notes when an agent should realize that their estimates could surely do a better job representing the actual values of some variables. To paraphrase \cite{Pettigrew2016}, this coherence notion assesses the \emph{epistemic} role of estimates---i.e. their role in representing facts about which we are uncertain---but it has nothing to say about their \emph{pragmatic} role---i.e. their role in guiding our actions.\footnote{Note that this non-pragmatic approach differs from the `depragmatized Dutch book arguments' of \cite{Howson2006} and \cite{Christensen1996}, which require agents to accept certain bets as `fair'; \cite{Maher1997} reconstructs and criticizes these authors' arguments. I do not invoke any notion of `fairness', or consider literal bets at all.}

Construed as a non-pragmatic constraint, sufficient coherence avoids the usual criticisms leveled against the classic `Dutch book' argument. If the classic argument is meant to justify the pragmatic role of credences, then it faces has a host of well-known problems. For one: by identifying credence with willingness to bet, the argument cannot account for risk-averse agents.\footnote{I take this concern to summarize a number of salient criticisms made along these lines due to \cite{Kaplan1998}.} For another: the argument does not account for cases where we cannot learn the actual truth values. This latter point reveals that a `Dutch book' strategy would fail spectacularly for justifying the pragmatic role of estimates for chances, as actual chance values may often be inaccessible.\footnote{As \cite{Gell-Mann2012} and \cite{Feintzeig2017} put it, the relevant bets are not `settleable'.}

I will not address the pragmatic role of estimates for chances presently. But it seems clear that regardless of how you \emph{use} an approximate calculation of chance, \emph{its goodness as an approximation} should be measured by how well it represents the values that you are unsure about. Intuitively, then, your estimates should rationally cohere just so long as you cannot guarantee that other estimates would do a better job playing their epistemic role. Sufficient coherence captures this intuition. Explicitly, I propose the following general answer to the functionalist's question:
\begin{quote}
    \textbf{The functionalist's answer.} Estimates for chances are rationally coherent if and only if they are sufficiently coherent.
\end{quote}
To parlay this answer into a defense of probabilism for stochastic theories, I will show that if one of several sets of states on a $C^*$-algebra is taken to yield the set of possible points, then estimates are sufficiently coherent if and only if they are given by some state on that algebra.

\section{Defending probabilism for stochastic theories}\label{sec:coherence}

This section contains the main technical results and uses them to defend probabilism for stochastic theories. In Section \ref{sec:coherence_games}, I discuss how sufficient coherence and full coherence are related to strongly separating hyperplanes, and I give a topological characterization of sufficient coherence in infinite dimensions. In Section \ref{sec:coherence_quantum}, I use this characterization to defend probabilism for stochastic theories. Finally, in Section \ref{sec:coherence_relating}, I characterize when sufficient coherence and full coherence align, and I show that there is a quantum state giving estimates that are sufficiently coherent, but not fully coherent.

\subsection{Sufficient coherence, full coherence, and strongly separating hyperplanes}\label{sec:coherence_games}

In this section, I relate sufficient coherence to the more familiar notion of full coherence, and in the process I arrive at a topological characterization of sufficient coherence in infinite dimensions. This characterization is based on de Finetti's insight that both coherence notions concern the existence of separating hyperplanes.

De Finetti's version of the `Dutch book' betting game in \emph{Theory of Probability} handily applies to any estimate function $e$ for any set of real-valued variables $K$ with possible points $\mathcal{Q}$.\footnote{De Finetti allows for variables take on complex values, as well, but does not explicitly treat this case in his proofs (\citeyear[p. 25, pp. 76--78]{deFinetti-ToP:2017}). It seems quite clear that his arguments should have a natural extensions to complex-valued random variables. That this extension is not immediate is due to the fact that assessment of coherence (as we will see) turns out to rely on hyperplane separation theorems from real analysis. Thus, some (uncomplicated, but subtle) translations are in order. Nonetheless: to keep things straightforward, I will stick to real numbers.} The game imagines a bookie, who picks some finite subset $\Gamma$ of $K$, and assigns dollar-value stakes to each variable. Let us treat this assignment as a \emph{stake function} $s: \Gamma\to \mathbb{R}$. Now you learn the following rule: if you choose play the game for some event $k$, the bookie will pay out (or collect, if $s(k)$ is negative) a dollar amount determined by $s(k)\cdot q(k)$ (presuming $q\in\mathcal{Q}$ gives the actual value of $k$). Given this, and given their expectation $e$, we suppose that the agent is willing to put down (or receive from the bookie, if $s(k)$ is negative) a dollar amount determined by $s(k) \cdot e(k)$ in order to play the game for $k$. We can formally describe this game with the following function:
    \begin{equation*}
    \mathcal{D}_{s}(q,e) := \sum_{k \in \Gamma     } s(k) [ q(k) - e(k)].
    \end{equation*}
Your estimate function \emph{admits a `Dutch book'} if the bookie can pick some finite set of variables and some stake function such that you are \emph{guaranteed} to lose money in this game.
\begin{quote}
    \textbf{`Dutch book'.} An estimate function $e$ \emph{admits a `Dutch book'} if and only if there is some finite subset $\Gamma$ of $K$ and stake function $s$ such that
    \begin{equation*}
    \text{for all } q\in\mathcal{Q}, \quad \mathcal{D}_{s}(q,e) < 0.
    \end{equation*}
\end{quote}
Full coherence simply requires that agents avoid `Dutch books'.
\begin{quote}
    \textbf{Full coherence.} An estimate function $e$ is \emph{fully coherent} if and only if it does not admit a `Dutch book'.
\end{quote}
In de Finetti's early work (\citeyear{deFinetti1931,deFinetti-F:1937}), the story ends here. In \emph{Theory of Probability}, however, he deems the notion of full coherence too strong, and he replaces it with sufficient coherence. 

De Finetti first tackles sufficient coherence by very slightly weakening the definition of full coherence. He reasons that there is nothing practically to be gained by eliminating a sure loss, if this loss can be made arbitrarily small. Thus, instead of requiring that agents avoid a sure loss, he merely requires that they can always get \emph{arbitrarily close to breaking even}. If an agent \emph{cannot} get arbitrarily close to breaking even, we say their estimate function \emph{admits a strong `Dutch book'}:
\begin{quote}
    \textbf{Strong `Dutch book'.} An estimate function $e$ \emph{admits a strong `Dutch book'} if and only if there is some finite subset $\Gamma$ of $K$, stake function $s$, and $\epsilon > 0$ such that
    \begin{equation*}\label{eq:dutch_def}
    \text{for all } q\in\mathcal{Q}, \quad \mathcal{D}_{s}(q,e) < -\epsilon.\footnote{This is de Finetti's `first criterion' for sufficient coherence (\citeyear[\S3.3]{deFinetti-ToP:1974}).}
    \end{equation*}
\end{quote}
It turns out that an estimate function is sufficiently coherent just in case it does not admit a strong `Dutch book'. To see why---and to obtain a topological characterization of sufficient coherence in $\mathbb{R}^K$---it helps to note that `Dutch books' are deeply related to separating hyperplanes, a familiar object of study in functional analysis.

Indeed, it is not hard to see that the existence of a strong `Dutch book' aligns with the existence of a \emph{strongly} separating hyperplane. A \emph{hyperplane} $[f=\alpha]$ refers to a set of points whose image under the non-zero linear functional $f$ is the real number $\alpha$. It will also be useful, shortly, to refer to $[f\leq\alpha]$ and $[f\geq\alpha]$ (resp. $[f < \alpha]$ and $[f > \alpha]$), the closed (resp. open) half spaces of such a hyperplane; these contain all points whose image under $f$ satisfies the specified relationship to $\alpha$. Finally, a hyperplane $[f=\alpha]$ \emph{strongly separates} a point from a set if there is some $\epsilon>0$ such that the point lies in $[f\leq\alpha]$ and the set lies in $[f\geq\alpha+\epsilon]$ (or vice versa). So now (recalling that $\mathbb{R}^\Gamma$ is self-dual) we have the following:
\begin{quote}
    \textbf{Strongly separating hyperplane.} Let $\mathcal{Q}|_{\Gamma}$ be the set of all restrictions of functions in $\mathcal{Q}$ to the finite subset $\Gamma$ of $K$. There is a \emph{strongly separating hyperplane} between a point $e|_\Gamma$ and the set $\mathcal{Q}|_\Gamma$ in $\mathbb{R}^\Gamma$ if and only if there is some $s \in \mathbb{R}^\Gamma$, $\alpha\in\mathbb{R}$, and $\epsilon > 0$  such that
    \begin{equation*}\label{eq:hyperplane}
    \text{for all } q \in \mathcal{Q}, \quad s     \cdot q|_\Gamma \leq \alpha \quad     \text{and} \quad s \cdot e|_\Gamma \geq     \alpha+\epsilon.
    \end{equation*}
\end{quote}
And so $e$ admitting a strong `Dutch book' is equivalent to it being strongly separated (by some hyperplane) from the set of possible points $\mathcal{Q}$ over some finite set of variables $\Gamma$.

This happy coincidence allows us to make full use of several powerful separating hyperplane theorems. The most useful of these has the following consequence: if a nonempty closed convex subset of $\mathbb{R}^n$ does not contain $e$, then there is a hyperplane strongly separating this subset from $e$.\footnote{This follows from a more general strong separating hyperplane theorem, for which see, e.g., \cite[p. 207]{Aliprantis2006}.}

De Finetti notes that this result allows us to topologically characterize, in finite dimensions, when an estimate function admits a strong `Dutch book'. First, let $\mathcal{Q}|_\Gamma$ be the set of all restrictions of functions in $\mathcal{Q}$ to the finite subset $\Gamma$ of $K$. Then denote the smallest closed convex set containing $\mathcal{Q}|_\Gamma$ by $\overline{\mathrm{co}}(\mathcal{Q}|_\Gamma)$ (this is the \emph{closed convex hull} of $\mathcal{Q}|_\Gamma$). On the one hand: if $e$ is contained in $\overline{\mathrm{co}}(\mathcal{Q}|_\Gamma)$, then it cannot be strongly separated from $\mathcal{Q}$ by some hyperplane. Suppose it were: invoke the well-known fact that any closed, convex set in $\mathbb{R}^{\Gamma}$ is the intersection of all closed half spaces that contain it, and note that this fact together with the supposition yields that $e$ both is and is not in $\overline{\mathrm{co}}(\mathcal{Q}|_{\Gamma})$! On the other hand: if $e$ is not in $\overline{\mathrm{co}}(\mathcal{Q}|_\Gamma)$, then by the aforementioned theorem, there \emph{is} such a strongly separating hyperplane. Thus, we immediately get the following characterization of strong `Dutch books': $e$ does not admit a strong `Dutch book' if and only if every finite restriction $e|_{\Gamma}$ lies in the closed convex hull of $\mathcal{Q}|_{\Gamma}$.\footnote{I've cleaned up some of de Finetti's terminology here. He writes: `Geometrically, a point $e$ represents a [sufficiently] coherent expectation if and only if there exists no hyperplane separating it from the set $\mathcal{Q}$ of possible points; this characterizes the points of the convex hull', where this is meant to be demonstrated for the case where $K$ is finite (\citeyear[p. 76]{deFinetti-ToP:2017}). By `convex hull' de Finetti means closed convex hull, and by `separating' de Finetti means strongly separating.}

De Finetti goes on to note that similar considerations reveal the equivalence of sufficient coherence and not admitting a strong `Dutch book'. And so the above topological characterization singles out the sufficiently coherent estimate functions.
\begin{claim}[\citealt{deFinetti-ToP:1974}]
An estimate function $e$ is sufficiently coherent if and only if it does not admit a strong `Dutch book'---that is, if and only if for every finite subset $\Gamma$ of $K$, $e|_\Gamma \in \overline{\mathrm{co}}(\mathcal{Q}|_\Gamma)$.
\end{claim}
\begin{proof}
First, suppose that $e$ is not sufficiently coherent because $e^*$ is closer to all points in $\mathcal{Q}|_\Gamma$, and then note that $e|_\Gamma$ cannot be in $\overline{\mathrm{co}}(\mathcal{Q}|_{\Gamma})$. Suppose it were: consider the bisecting normal hyperplane, $[d(e|_\Gamma,~\cdot~) -d(e^*|_\Gamma,~\cdot~)= 0]$. Note $e^*|_{\Gamma}$ is closer (in Euclidean distance) to $q|_\Gamma$ than $e|_{\Gamma}$ if and only if $q|_\Gamma$ lies in the open half-space containing $e^*|_{\Gamma}$. But since this is the case, the closed half space containing $e^*|_{\Gamma}$ also contains $\overline{\mathrm{co}}(\mathcal{Q}|_{\Gamma})$, and so $e|_\Gamma$ is \emph{not} in that set!

Second, suppose that $e$ \emph{is} sufficiently coherent, and note that $e|_\Gamma$ \emph{must} be in $\overline{\mathrm{co}}(\mathcal{Q}|_{\Gamma})$. Suppose it weren't: then there is a hyperplane $[f=\alpha]$ strongly separating it from this set. Without loss of generality, let $[f=\alpha]$ be closer to $e|_\Gamma$ than  $[f=\alpha+\epsilon]$, and let $e^*|_\Gamma$ be the point such that $[f=\alpha]$ is the normal hyperplane bisecting the line segment joining $e|_\Gamma$ and $e^*|_\Gamma$: clearly, the latter point is closer to all points $q|_\Gamma$.

So an expectation is sufficiently coherent just in case it lies in the closed convex hull of possible points over every finite set of variables---and so just in case it does not admit a strong `Dutch book'.\footnote{A version of this argument appears in \cite[\S3.4]{deFinetti-ToP:1974,deFinetti-ToP:2017}; its formulation as a non-pragmatic argument for classical probabilism is due to \cite{Joyce1998}. \cite{Williams2012,Williams2012a} offers a similar argument (and illuminating illustrations) for the case where it is assumed that the set of possible points $\mathcal{Q}$ is finite (Williams refers to this set as the `set of worlds').}
\end{proof}{}
It turns out that this finite-dimensional characterization extends to a topological characterization in $\mathbb{R}^K$, for $K$ possibly uncountably infinite (and where $\mathbb{R}^K$ is given the usual topology of pointwise convergence). This extension hinges on the fact that $e$ lies in the closed convex hull of $\mathcal{Q}$ if and only if each of its finite restrictions lies in the closed convex hull of the appropriate finite restriction of $\mathcal{Q}$. (I relegate the proof of this fact to the Appendix.)
\begin{prop}\label{prop:finite_res}
Let $A$ be a nonempty subset of $\mathbb{R}^K$. For $e$ a function in this space, $e\in\overline{\mathrm{co}}(A)$ if and only if $e|_{\Gamma}  \in \overline{\mathrm{co}}(A|_{\Gamma})$ for all finite $\Gamma\subseteq K$.
\end{prop}
\noindent
An immediate consequence of this proposition and the above claim is that $e$ is sufficiently coherent if and only if it lies in the closed convex hull of $\mathcal{Q}$ in $\mathbb{R}^{K}$.
\begin{prop}\label{prop:suff_res}
An estimate function $e\in \mathbb{R}^{K}$ is sufficiently coherent if and only if $e \in \overline{\mathrm{co}}(\mathcal{Q})$.
\end{prop}
\noindent
This characterization of sufficient coherence allows for a straightforward defense of probabilism for stochastic theories.

\subsection{Using sufficient coherence to defend probabilism for stochastic theories}\label{sec:coherence_quantum}

In this section, I apply the foregoing characterization of sufficient coherence to the algebraic states defined in the $C^*$-algebraic approach to stochastic theories, for several different stipulations of possible assignments of chance values. The algebraic approach affords a high degree of generality: in most quantum and classical theories, measurement events are (or can be) represented as elements in some $C^*$-algebra or another. So to start, I will briefly review this approach.

Recall that a (concrete) $C^*$-algebra $\mathfrak{A}$ is a self-adjoint, unital, and norm-closed subalgebra of the algebra $\mathcal{B(H)}$ of bounded operators on a Hilbert space. The effects $\mathcal{E(\mathfrak{A})}$ are the positive elements of $\mathfrak{A}$ that are dominated by the identity. A subset of these, the projections $\mathcal{P(\mathfrak{A})}$, are the (positive) idempotent elements of the algebra. Any effect in a $C^*$-algebra may be interpreted as a  measurement event. I will not review this interpretation in depth here; for more details, see \cite{Busch1995}. Suffice to say that, in general, an effect may be associated with observing a specific outcome for an unsharp measurement (a POVM); as a special case, a projection may be associated with observing a specific outcome for a sharp measurement (a PVM). I will refer to the set of effects (or the set of projections) as \emph{a set of measurement events}:
\begin{quote}
    \textbf{Measurement events.} $K$ is \emph{a set of measurement events} if it is either the set $\mathcal{P}(\mathfrak{A})$ or the set $\mathcal{E}(\mathfrak{A})$ for some $C^*$-algebra $\mathfrak{A}$. If the $C^*$-algebra is defined on a finite-dimensional Hilbert space, then $K$ is \emph{a set of finite measurement events}.
\end{quote}

A \emph{state} $\omega$ on a $C^*$-algebra $\mathfrak{A}$ is a positive and normalized linear functional on $\mathfrak{A}$; I refer to states as `algebraic states' when I wish to flag their association with a $C^*$-algebra. By their positivity, algebraic states give effects values in the unit interval, and by their linearity, they are finitely additive over the effects (although they need not be countably or completely additive over them). An extremal element of the convex set $\mathcal{S}$ of states---that is, a state that cannot be written as a non-trivial convex combination of at least two different states---is called a \emph{pure state}. The set $\mathcal{S}$ is often called the \emph{state space} when it is treated as a subset of the dual $\mathfrak{A}^*$ of the algebra equipped with the weak$^*$ topology---this is the topology of pointwise-convergence for positive linear functionals on $\mathfrak{A}$. This approach is quite powerful, as the state space $\mathcal{S}$ turns out to be compact in this topology (c.f. \citealt[Theorem 2.3.15]{BratteliRobinson:vol1}).

In order to recover the textbook orthodoxy for density operators qua proper mixtures, we will also want to define vector states and normal states. A \emph{vector state} $\omega_x$ is a state such that $\omega_x (~\cdot~) = \langle x, ~\cdot~ x \rangle$ for some normalized $x\in\mathcal{H}$. A \emph{normal state} $\omega_{\rho}$ is a state such that $\omega_{\rho}(~\cdot~) = \mathrm{Tr}(\rho ~\cdot~)$ for some density operator $\rho \in\mathcal{B(H)}$ (and so it is some countable convex combination of vector states). I will note that normal states are typically defined in the framework of von Neumann algebras---that is, $C^*$-algebras that are also closed in the weak-operator topology---as the states on such algebras that have a certain weak-operator continuity property. But for concrete von Neumann algebras, this latter definition of normal states coincides with my own (c.f. \citealt[Theorem 7.1.12]{KadisonRingrose:vol2}).

That is all the background we will need. I will now introduce a couple of short-hand descriptions of points in $\mathbb{R}^K$ to help state the main results.
\begin{quote}
    \textbf{Points given by states on $C^*$-algebras.} Let a point in $\mathbb{R}^K$ be \emph{pure}, \emph{vectorial}, \emph{normal}, or \emph{algebraic} if it is the restriction of a pure state, a vector state, a normal state, or a state on $\mathfrak{A}$ to $K$, respectively.
\end{quote}
Textbook orthodoxy for density operators qua proper mixtures illustrates one case where vectorial points ought to be treated as the possible assignments of chance values---these points encode the chances prescribed by wavefunctions. Other foundational or interpretive approaches to quantum theories may take one of the other three sets to yield the possible assignments of chance values. For example, one may use the `transcendental' argument of \cite{Arageorgis2017} to motivate taking normal states to yield the possible points: one part of this argument suggests that, since we think we can reliably prepare normal states in the lab, and since neither Sch\"{o}dinger evolution nor von Neuman projection (i.e. the usual ways of describing the evolution of states) can take a non-normal state to a normal state, we should think that only normal states correspond to possible assignments of chance values.\footnote{See also Chen's (\citeyear{Chen2017}) discussion of a variety of approaches taking density operators to be fundamental, including a such an approach to the Ghirardi-Rimini-Weber collapse theory.}





So long as you choose one of the above four sets of points to be the set of possible assignments of chance values, then the application of sufficient coherence is straightforward. We need only a basic property of the state space $\mathcal{S}$---namely, that this space is the closed convex hull of the set of pure states, the set of vector states, the set of normal states, and the set of states. But note that because we have restricted our attention to the set of measurement events, a \emph{subset} of the set of elements of a $C^*$-algebra, this nice property of the state space is not immediately applicable to estimates. The compactness of the state space is crucial for showing that an analogous property holds for the set of algebraic points in $\mathbb{R}^K$. That work done, Proposition \ref{prop:suff_res} suffices to show the following:

\begin{theorem}\label{thm:gen}
Let $K$ be a set of measurement events. Suppose the set $\mathcal{Q} \subseteq \mathbb{R}^K$ of possible points is one of:
\begin{itemize}
    \item[(i)] the set of pure points;
    \item[(ii)] the set of vectorial points;
    \item[(iii)] the set of normal points; or
    \item[(iv)] the set of algebraic points.
\end{itemize}
Then an estimate function is sufficiently coherent iff it is algebraic.
\end{theorem}
\begin{proof}
Let $\mathfrak{A}$ be the concrete $C^*$-algebra associated with $K$, let $\mathcal{S}$ be the set of states in $\mathfrak{A}^*$ with the weak$^*$ topology. Let $\mathcal{Q}_\omega \subseteq \mathcal{S}$ be the set of states that restrict to possible points; i.e. $\mathcal{Q}_\omega|_K=\mathcal{Q}$. Given any of our four suppositions,
\begin{equation}\label{eq:closedhull}
\mathcal{S} = \overline{\mathrm{co}}(\mathcal{Q}_\omega).
\end{equation}
If (i) $\mathcal{Q}_\omega$ is the set of pure states: equation (\ref{eq:closedhull}) follows from, e.g., Theorem 2.3.15 of \cite{BratteliRobinson:vol1}. If (ii) $\mathcal{Q}_\omega$ is the set of vector states: equation (\ref{eq:closedhull}) follows from, e.g., Corollary 4.3.10 of \cite{KadisonRingrose:vol1}, since $\mathcal{A}$ is a unital, self-adjoint subspace of some $\mathcal{B(H)}$. If (iii) $\mathcal{Q}_\omega$ is the set of normal states: now letting $\mathcal{V}$ be the vector states, $\mathrm{co}(\mathcal{V})\subseteq\mathrm{co}(\mathcal{Q}_\omega)=\mathcal{Q}_\omega\subseteq \overline{\mathrm{co}}(\mathcal{V})=\mathcal{S}$ and so equation (\ref{eq:closedhull}) follows from the previous. If (iv) $\mathcal{Q}_\omega$ is the set of states: equation (\ref{eq:closedhull}) is trivial.

Now note that $\mathcal{S}|_K = \overline{\mathrm{co}}(\mathcal{Q}_\omega)|_K = \overline{\mathrm{co}} (\mathcal{Q})$; the latter equality follows directly from the compactness of $\mathcal{S}$ (and the fact that $\mathbb{R}^K$ has the topology of pointwise convergence). 
Proposition \ref{prop:suff_res} completes the proof.
\end{proof}
\noindent
Assuming you agree with the functionalist's answer (c.f. Section \ref{sec:prelim_coherence}), then Theorem \ref{thm:gen} yields a defense of probabilism for stochastic theories on four different assumptions regarding which states describe possible assignments of chance values. Explicitly:
\begin{quote}
    \textbf{Probabilism for stochastic theories (restated).} Suppose that (i) pure states, (ii) vector states, (iii) normal states, or (iv) states on a given $C^*$-algebra yield the possible assignments of chance values. Then estimates for chances are rationally coherent if and only if they are given by a state on that algebra.
\end{quote}
When possible assignments of chance values are given by (i) pure states, the thesis vindicates textbook orthodoxy for density operators qua proper mixtures in finite dimensions (as all states are normal in this case).

I conclude this section with a brief consideration of some technical and philosophical implications of Theorem \ref{thm:gen} for the usual quantum logics of effects and of projections. These implications are straightforward, and they demonstrate links between Theorem \ref{thm:gen} and algebraic versions of Gleason's theorem. I will consider such versions of Gleason's theorem for effects and for projections in turn.

First, let $K$ be the set of effects in some $\mathcal{B(H)}$. Busch's (\citeyear{Busch2003}) analog of Gleason's theorem for POVMs implies that the \emph{Kolmogorovian} functions in $\mathbb{R}^K$---that is, the finitely additive functions valued in the unit interval and assigning one to the identity---are all and only the restrictions of states on $\mathcal{B(H)}$ to $K$ (for $\mathcal{H}$ of \emph{any} countable dimension). Thus, if the possible assignments of chance values are given by any of (i)--(iv), then an estimate function is sufficiently coherent if and only if it is Kolmogorovian.

Second, let $K$ be the set of projections in some von Neumann algebra. For \emph{most} von Neumann algebras, the Kolmogorovian functions in $\mathbb{R}^K$ are all and only the restrictions of states on the algebra to the projections. A sufficient condition for this result to hold is that the von Neumann algebra in question has no direct summand of type $I_2$---for an explication of this concept and a proof of the result, see \cite{Hamhalter2013}.\footnote{See also the lucid discussion of \cite{Redei2007}, who present the result as it appears here.} For all von Neumann algebras of this sort, if the possible assignments of chance values are given by any of (i)--(iv), then an estimate function is sufficiently coherent if and only if it is Kolmogorovian. Moreover: suppose that (ii) vector states yield the possible assignments of chance values, and suppose that these assignments correspond to many-valued `truth valuations'. Now, the standard Tarskian definition of entailment recovers the usual ordering of projections, providing a clear physical semantics for projections qua either an ordered partial algebra or a Hilbert lattice.\footnote{Explicitly, take a possible assignment of chance values $q$ to be a truth valuation in the sense that $q(a)=0$ means $a$ is false, $q(a)=1$ means $a$ is true, and $q(a)\in (0,1)$ means $a$ has some degree of truth. The standard Tarskian definition of entailment says that $a$ implies $b$ just in case $q(a)=1$ implies $q(b)=1$ for all $q\in\mathcal{Q}$.}

We have seen that sufficient coherence affords several defenses of probabilism for stochastic theories, and these defenses have a wide range of applications. But the reader partial to the traditional `Dutch book' argument may wonder if we \emph{need} sufficient coherence (or strong `Dutch books') to carry out these defenses. Can we get by with full coherence alone? Occasionally, we can. For example, full coherence can be used to defend a probabilist thesis for finite-dimensional quantum systems. But it turns out that, if we wish to defend the coherence of any finitely additive quantum probability (qua estimates for chances) in an \emph{infinite}-dimensional system, then sufficient coherence may be \emph{necessary}.

\subsection{Characterizing the difference between full coherence and sufficient coherence}\label{sec:coherence_relating}

In this section, I characterize when full coherence and sufficient coherence align, and I show that sufficient coherence is needed for one version of probabilism for stochastic theories. This section extends the work of \cite{Williams2012,Williams2012a} relating accuracy-dominance to `Dutch books' for finite sets of possible points, as well as the work of \cite{Paris2001} and \cite{Kuhr2007} on `Dutch book' theorems for non-classical and many-valued logics.

I first introduce a couple of definitions, lightly adapted from the standard notion of an exposed point in functional analysis, in order to make the characterization straightforward.\footnote{For an overview of exposed points, see \cite[Section 7.15]{Aliprantis2006}.}
\begin{quote}
\textbf{Relatively exposed point.} Let $A$ be a nonempty convex subset of $\mathbb{R}^K$. Say a point $e\not \in A$ is \emph{exposed relative to $A$} if $e$ is the unique maximizer (or minimizer) over $A\cup\{e\}$ of a non-zero linear functional.
\end{quote}
\noindent
Note that this definition is equivalent to saying that there is a hyperplane containing $e$ in a closed half space and leaving $A$ in the disjoint open half space. For $K$ finite, the definition is also equivalent to saying that there is some $s \in \mathbb{R}^K$ such that $s\cdot a < s\cdot e$ for all $a\in A$.
\begin{quote}
\textbf{Finitely convexly protected.} Let $A$ be a nonempty subset of $\mathbb{R}^K$. Say $A$ is \emph{finitely convexly protected} if, for all finite $\Gamma \subseteq K$, there is no point $b \in \overline{\mathrm{co}}(A|_{\Gamma})$ that is exposed relative to $\mathrm{co}(A|_{\Gamma})$ and that extends to some element $a\in \overline{\mathrm{co}}(A)$ (i.e. $a|_\Gamma = b$).
\end{quote}
\noindent
We now have the following:
\begin{prop}\label{prop:full_suff}
Suppose $\mathcal{Q}$ is finitely convexly protected. Then an estimate function $e\in \mathbb{R}^K$ is fully coherent if and only if $e \in \overline{\mathrm{co}}(\mathcal{Q})$.
\end{prop}
\begin{proof}
Suppose that $e\in \overline{\mathrm{co}}(\mathcal{Q})$, and suppose towards a contradiction that $e$ admits a `Dutch book'. Then for some finite $\Gamma\subseteq L$, $e|_{\Gamma}$ is exposed relative to $\mathrm{co}(\mathcal{Q}|_{\Gamma})$. But then note that by Proposition \ref{prop:finite_res}, $e|_{\Gamma}\in\overline{\mathrm{co}}(\mathcal{Q}|_{\Gamma})$. So then by our supposition, $e|_{\Gamma}$ cannot be exposed relative to $\mathrm{co}(\mathcal{Q}|_{\Gamma})$.

For the converse, suppose that $e \not\in \overline{\mathrm{co}}(\mathcal{Q})$. By our supposition and Proposition \ref{prop:finite_res}, there is some finite $\Gamma\subseteq L$ such that $e|_{\Gamma} \not\in \mathrm{co}(\mathcal{Q}|_{\Gamma})$. Treat $e|_{\Gamma}$ as a point in $\mathbb{R}^{\Gamma}$ and treat $\mathrm{co}(\mathcal{Q}|_{\Gamma})$ as a closed convex set in $\mathbb{R}^{\Gamma}$. There is a hyperplane strongly separating $e|_{\Gamma}$ and $\mathrm{co}(\mathcal{Q}|_{\Gamma})$---and so $e$ admits a `Dutch book'.
\end{proof}
\noindent
This proposition is a generalization (and a much simpler proof) of the result of K\"{u}hr and Mundici (\citeyear[Theorem 2.3]{Kuhr2007}), who rely on the stronger supposition that $\mathcal{Q}|_{\Gamma}$ is closed for all finite $\Gamma\subseteq K$. Moreover: the antecedent of Proposition \ref{prop:full_suff} is also a \emph{necessary} condition for $\overline{\mathrm{co}}(\mathcal{Q})$ to contain all and only those estimate functions that are fully coherent.

\begin{prop}\label{prop:full_nec}
Suppose that $\mathcal{Q}$ is not finitely convexly protected. Then there is an estimate function $e\in \overline{\mathrm{co}}(\mathcal{Q})$ that is not fully coherent.
\end{prop}
\begin{proof}
Let $\Gamma$ be the finite subset of $K$ such that $\mathrm{co}(\mathcal{Q}|_{\Gamma})$ has a relatively exposed closure point $e|_{\Gamma}$ that extends to some element $e\in \overline{\mathrm{co}}(\mathcal{Q})$. Since $e|_{\Gamma}$ is a relatively exposed closure point of $\mathrm{co}(\mathcal{Q}|_{\Gamma})$, there is some $s\in\mathbb{R}^{\Gamma}$ such that $s\cdot f < s\cdot e|_{\Gamma}$ for all $f\in \mathrm{co}(\mathcal{Q}|_{\Gamma})$. Thus, $e$ admits a `Dutch book'.
\end{proof}
\noindent
So the condition of being finitely convexly protected characterizes those sets of possible points $\mathcal{Q}$ that identify $\overline{\mathrm{co}}(\mathcal{Q})$ as the set of fully coherent estimate functions. Thus, we see that full coherence and sufficient coherence align precisely when $\mathcal{Q}$ is finitely convexly protected.

When full coherence and sufficient coherence are used to defend classical probabilism, they are clearly equivalent. That is: if the variables (e.g. propositions) are $\{0,1\}$-valued (e.g. letting 0 be false and 1 be true), then $\mathcal{Q}|_{\Gamma}$ is finite and so closed (and so compact) for each finite $\Gamma \subseteq K$. Thus, each $\mathrm{co}(\mathcal{Q}|_{\Gamma})$ is closed (by Carath\'{e}odory's theorem) and so $\mathcal{Q}$ is (trivially) finitely convexly protected.

But if we ask after estimates for chances, then full coherence and sufficient coherence will generally \emph{not} be equivalent, as $\mathcal{Q}$ no longer need be finitely convexly protected. In fact, recall the aforementioned textbook orthodoxy for density operators qua proper mixtures: this provides a natural example where $\mathcal{Q}$ fails to have the property.

If we assume $\mathcal{H}$ is finite-dimensional, then $\mathcal{Q}$ \emph{is} finitely convexly protected on this approach, and we \emph{can} still use full coherence to defend it. That is, sufficient coherence and full coherence align in this case:
\begin{prop}\label{prop:finite}
Let $K$ be a set of finite measurement events, and let $q\in\mathbb{R}^K$ be possible iff it is vectorial (i.e. iff it is pure). An estimate function is fully coherent iff it is algebraic (i.e. iff it is normal).
\end{prop}
\begin{proof}
Let $\mathfrak{A}\subseteq \mathcal{B(H)}$ be the $C^*$-algebra associated with $K$ (where $\mathcal{H}$ is finite-dimensional). We wish to show that $\mathcal{Q} = \mathcal{V}|_K$ (where $\mathcal{V}$ is the set of vector states in $\mathfrak{A}^*$) is finitely convexly protected. To do this, it suffices to establish the stronger condition that $\mathcal{Q}$ is closed.

To this end, note that the functions in $\mathcal{V}$ extend to the entire space of operators. Let $\mathcal{V}'$ be the vector states on $\mathcal{B(H)}$; $\mathcal{V}'$ is weakly$^*$ closed in the state space of $\mathcal{B(H)}$ (\citealt[4.6.18 and 4.6.67]{KadisonRingrose:vol1}); since $\mathcal{S}$ is weakly$^*$ compact, $\mathcal{V}'$ is also weakly$^*$ compact. Thus, $\mathcal{V}'|_{K} = \mathcal{V}|_{K} = \mathcal{Q}$ is closed.
\end{proof}
\noindent
Here, every algebraic state is given by some density operator. But when we pass to infinite dimensions, Theorem \ref{thm:gen} tells us that we need not require an estimate function to match some normal state. Non-normal states will do the job just fine---the resulting estimates will still be sufficiently coherent.

However, these estimates need no longer be fully coherent! To see this, it suffices to note that there exists an effect that is compact and that has a trivial kernel---by this latter fact, no vector state will assign zero to this effect. However, it is well-known that when $\mathcal{H}$ is infinite-dimensional, there are pure states on the set of bounded operators that assign zero to all the compact operators; such a state will be exposed relative to the state space, when all points are restricted to the effect in question. And this is all we need to show that $\mathcal{Q}$ is not finitely convexly protected.

\begin{prop}\label{prop:not_protected}
Let $K$ be the set of effects in $\mathfrak{A}=\mathcal{B(H)}$ for $\mathcal{H}=\ell^2$, and let $q\in\mathbb{R}^K$ be possible iff it is vectorial. There is an algebraic point that is not fully coherent (i.e. $\mathcal{Q}$ is not finitely convexly protected).
\end{prop}
\begin{proof}
Let $\mathcal{H} = \ell^2$ and let $E$ be the operator 
$$
E(x_i)_{i=1}^{\infty} := (x_i / i)_{i=1}^{\infty}.
$$
$E$ is compact and its kernel is trivial (\citealt[Ex. 8.1-6]{Kreyszig1978}). Moreover: $E$ is positive and dominated by the identity, so it is an effect, i.e. $E\in \mathcal{E}(\mathfrak{A})$. So: for any vector state $\omega_x$, $0 < \omega_x(E) \leq 1$. However: since $\dim (\mathcal{H}) = \infty$, there is a pure state that is 0 on the set of compact operators \cite[Ex. 4.6.69]{KadisonRingrose:vol1}; denote this state by $\omega$. Clearly, for $\Gamma = \{ E \}$, the point $\omega|_{\Gamma}$ is exposed relative to $\mathrm{co}(\mathcal{Q}|_{\Gamma})$, and so $\mathcal{Q}$ is not finitely convexly protected (and, in particular, $\omega|_K$ is not fully coherent).\footnote{As an anonymous reviewer notes, the mechanics of this proof suggest the following conjecture: when $K$ and $\mathcal{Q}$ are as described in Proposition \ref{prop:not_protected}, a sufficiently coherent estimate function is not fully coherent just in case it corresponds to a non-normal state. I leave an assessment of this conjecture for future work.} 
\end{proof}
\noindent
For $\mathcal{B(H)}$ and any other von Neumann algebra, the normal states are just those states that are completely additive. Thus, we \emph{need} sufficient coherence to defend the rational coherence of finitely-but-not-completely additive quantum probabilities qua estimates for chances.

\section{Conclusion}\label{sec:context}

I have shown how to give an accuracy-dominance defense of the rational coherence of estimates for chances that follow probabilistic rules. This defense uses a very thin view of objective probability---functionalism about chance---that posits only that chances are (at least) our best estimates for relative frequencies in repeated observations of a given sort. Functionalism is \emph{compatible} with accounts of objective probability that explain chances with natural properties, but it only \emph{requires} intersubjective agreement about chance values. Moreover, since agents can be uncertain about the sort of an observation, they can be uncertain about chances. This thin notion of chance allows for agents to estimate chance values---literally, to form estimates for best frequency estimates. And so de Finetti's accuracy-dominance constraints on rationally coherent estimates may be applied to estimates for chances.

This application yields a defense of \emph{probabilism for stochastic theories}, the thesis that estimates for chances rationally cohere if and only if they are given by a state on a $C^*$-algebra, when a special subset of states yields the possible assignments of chance values. In one case, this defense vindicates one orthodox use of density operators---viz. using them to describe uncertainty about which wavefunction describes the true state of a system. In sum, the application yields four notable results:
\begin{itemize}
    \item a topological characterization of when estimates for the values of (a possibly uncountably infinite number of) variables avoid accuracy-dominance, i.e. are \emph{sufficiently coherent} (namely, that an estimate function is sufficiently coherent just in case it lies in the \emph{closed convex hull} of $\mathcal{Q}$, the set of possible points; c.f. Proposition \ref{prop:suff_res});
    \item the result that, supposing that (i) pure states, (ii) vector states, (iii) normal states, or (iv) states on a given $C^*$-algebra yield the possible assignments of chance values, then estimates for chances are sufficiently coherent if and only if they are given by a state on that algebra (c.f. Theorem \ref{thm:gen});
    \item a characterization, for a possibly infinite set of possible points $\mathcal{Q}$, of when sufficient coherence aligns with the more familiar notion of full coherence, i.e. of avoiding `Dutch books' (namely, when $\mathcal{Q}$ is \emph{finitely convexly protected}; c.f. Propositions \ref{prop:full_suff} and \ref{prop:full_nec}); and
    \item the result that if vector states give the possible assignments of chance values, then there may be an estimate function yielded by an algebraic state that is sufficiently coherent but not fully coherent (c.f. Proposition \ref{prop:not_protected}).
\end{itemize}
Unexplored areas for further research include the relationship between `Dutch book' arguments for countable additivity and additivity conditions for states on $C^*$-algebras, as well as the relationship between diachronic `Dutch book' arguments and methods for updating estimates for chances. But if the above has given the reader a flavor of the sorts of defenses and applications possible for probabilism for stochastic theories, then that much will suffice for now.

\section*{Acknowledgements}

This work was supported in part by the National Science Foundation under Grant \#1734155.

\appendix

\section{Towards topologically characterizing sufficient coherence in infinite dimensions}
Recall that Proposition 1 was needed for the topological characterization of sufficiently coherent expectations in $\mathbb{R}^K$ with the product topology when $K$ is infinite. The proof of this proposition is included below for completeness. In the following, $\mathrm{Fin}(K)$ refers to the set of all finite subsets of $K$, ordered by inclusion.

\begin{proof}[Proof of Proposition 1] We aim to show that, for $X$ a subset of $\mathbb{R}^K$ (where this latter set is equipped with the product topology), $e\in\overline{\mathrm{co}}(X)$ if and only if $e|_{\Gamma}  \in \overline{\mathrm{co}}(X|_{\Gamma})$ for all $\Gamma\in \mathrm{Fin}( K)$. To demonstrate this, we show the following stronger claim:
$$
\text{for } A\subseteq \mathbb{R}^K,\quad e\in\overline{A} ~ \Longleftrightarrow ~   \forall\, \Gamma\in \mathrm{Fin}( K), ~  e|_{\Gamma}  \in \overline{A|_\Gamma}.
$$
(For $A=\mathrm{co}(X)$, note that $\mathrm{co}(X|_\Gamma)=\mathrm{co}(X)|_\Gamma$, and that the closed convex hull of a set is also the closure of the convex hull of that set.) Recall that a point belongs to the closure of a set if and only if it is a limit of a net in that set. Our strategy builds nets to witness the membership of the given points in the given closures.
\begin{enumerate}
    \item{$e\in\overline{A} ~ \Longrightarrow ~   \forall\, \Gamma\in \mathrm{Fin}( K), ~  e|_{\Gamma}  \in \overline{A|_{\Gamma}}$
    
    Suppose $e\in\overline{A} $. There is a net $n : D \to A :: d \mapsto a $ (for $D$ a directed set) that converges pointwise to $e$. For every $\Gamma\in \mathrm{Fin}( K)$, the net $n_{\Gamma}(d):= n(d)|_\Gamma$ converges pointwise to $e|_\Gamma$.}
    \item{$e\in\overline{A} ~ \Longleftarrow ~   \forall\, \Gamma\in \mathrm{Fin}( K), ~  e|_{\Gamma}  \in \overline{A|_{\Gamma}}$
    
    Suppose that, for all $\Gamma\in \mathrm{Fin}(K)$, $e|_{\Gamma}  \in \overline{A|_{\Gamma}}$. We have two cases: either (i) every finite restriction of $e$ belongs to the appropriate $A|_\Gamma$, or (ii) there is at least one that does not. For each case, we define a net $n$ in $A$ that converges pointwise to $e$.
    
    \begin{itemize}
        \item[(i)]{$\forall \, \Gamma\in \mathrm{Fin}( K), ~ e|_{\Gamma}  \in A|_{\Gamma}$
        
        For the first case, note that each $e|_{\Gamma}$ extends to some (possibly non-unique) point in $A$; for each $\Gamma$, label one of these points as $a_{\Gamma}$ (so $a_\Gamma|_\Gamma =e|_\Gamma$). The net
        \begin{equation*}
        n : \mathrm{Fin}(K) \to A :: \Gamma \mapsto a_{\Gamma}
        \end{equation*}
        pointwise converges to $e$. So $e\in\overline{A}$.}
        
        \item[(ii)]{$\exists\, \Gamma\in \mathrm{Fin}( K), ~ e|_{\Gamma}  \not\in A|_{\Gamma} $
        
        For the second case, we note that there are nets converging to $e|_\Gamma$ that do not contain these limits. We use this fact specify a subset of $A$ in which every pair of elements has an upper bound, and which allows for the definition of a net that converges pointwise to $e$ in $A$.

    Let $\Gamma_G \in\mathrm{Fin}(K)$ be a set such that $e|_{\Gamma_G}  \not\in A|_{\Gamma_G}$, and let $G$ be the set of all $\Gamma\in\mathrm{Fin}(K)$ that contain $\Gamma_G$ (i.e. such that $\Gamma \supseteq \Gamma_G $). Clearly, for all $\Gamma\in G$, $e|_\Gamma \not\in A|_\Gamma$.

    Moreover, for every $\Gamma\in G$, there is a net that converges pointwise (equivalently, uniformly) to $e|_{\Gamma}$ that does not contain it; for every such $\Gamma$, label one such net $n_\Gamma^*: D_{\Gamma}\to A|_{\Gamma}$ (for $D_\Gamma$ some directed set). Now note that each element in the range of $n_{\Gamma}^*$ extends to some (possibly non-unique) element $a\in A$. So associate with each net $n_{\Gamma}^*$ a net $n_{\Gamma}: D_{\Gamma}\to A$ such that $n_{\Gamma}(d)|_\Gamma = n_{\Gamma}^*(d)$.

    I will refer to the direction of each $D_{\Gamma}$ as $\succcurlyeq_{\Gamma}$. Now, for each pair $(\Gamma,d)$, define $E_{(\Gamma,d)}$ as the following set:
    $$
    E_{(\Gamma,d)} := \{ \epsilon : \epsilon >0 \text{ and } \forall d' \succcurlyeq_{\Gamma} d, \ \mathrm{max}_{k\in \Gamma}\left | n_{\Gamma}(d') (k) - e( k) 
    \right |  < \epsilon \}.
    $$
    By $n_\Gamma^*$ uniformly convergent, for every $\epsilon>0$, there is some $d \in D_\Gamma$ such that $\epsilon \in E_{(\Gamma,d)}$.

    Now define $D:=\{(\Gamma,d) : \Gamma \in G \text{ and } d \in D_{\Gamma}\}$, and let $n$ be the follwing map:
    \begin{equation*}
    n : D\to A :: (\Gamma,d ) \mapsto n_{\Gamma}(d).
    \end{equation*}
    Define the binary relation $\preccurlyeq$ on $D$ by the condition that $(\Gamma,d) \preccurlyeq (\Gamma',d')$ if and only if $\Gamma \subseteq \Gamma'$ and $E_{(\Gamma,d)} \subseteq E_{(\Gamma',d')}$. I will now show that $\preccurlyeq$ is a direction, making $n$ a net in $A$.

    It is clear that $\preccurlyeq$ is reflexive and transitive. To see that any two $(\Gamma,d),\, (\Gamma',d')$ have an upper bound, we can assume without loss of generality that $\Gamma \subseteq \Gamma'$. Note that $E_{(\Gamma,d)}$ is bounded from below by
    $$
    \epsilon_{\mathrm{min}} =: \mathrm{max}_{k\in \Gamma} \left | n_{\Gamma}(d) (k) - e(k) \right |
    $$
    (and it does not include $\epsilon_{\mathrm{min}}$), and it must be that $\epsilon_{\mathrm{min}} > 0$ by $e|_{\Gamma}\not \in A|_{\Gamma}$. But by $n_{\Gamma'}^*$ uniformly convergent to $e|_{\Gamma'}$, there is some $\beta'$ such that for all $\beta \succcurlyeq_{\Gamma'} \beta'$, $\mathrm{max}_{k\in \Gamma'}\left | n_{\Gamma'}(\beta) (k) - e(k) \right | < \epsilon_{\mathrm{min}}$. Since $\preccurlyeq_{\Gamma'}$ is a direction, there is some $ \eta$ such that $\beta' \preccurlyeq_{\Gamma'} \eta$ and $d' \preccurlyeq_{\Gamma'} \eta$. Note that $(\Gamma,d) \preccurlyeq (\Gamma', \eta)$ and $(\Gamma',d') \preccurlyeq (\Gamma', \eta)$. Thus, $\preccurlyeq$ is a direction, and so $n$ is a net.

    It is immediate that $n$ converges pointwise to $e$ (for every point in $K$, there is a $\Gamma\in G$ that contains it). Thus, $e\in\overline{A}$.}
    \end{itemize}}
\end{enumerate}

\end{proof}





\bibliographystyle{apa-good}
\bibliography{contextuality}

\end{document}